\documentclass[a4paper,pdftex,reqno,10pt]{amsart}

\usepackage[latin1]{inputenc}
\usepackage[T1]{fontenc}

\usepackage{amssymb,amsmath}
\usepackage{setspace}
\usepackage[normalem]{ulem}
\usepackage{color}

\usepackage[english]{babel}

\newtheorem{theorem}{Theorem}
\newtheorem{proposition}{Proposition}
\newtheorem{lemma}{Lemma}

\begin{document}
\title{The power of the largest player}

\author{Sascha Kurz}

\address{Department of Mathematics, University of Bayreuth, 95440 Bayreuth, Germany}
\email{sascha.kurz@uni-bayreuth.de}

\date{}

\begin{abstract}
Decisions in a shareholder meeting or a legislative committee are often modeled as a weighted game. 
Influence of a member is then measured by a power index. A large variety of different indices has been 
introduced in the literature. This paper analyzes how power indices differ with respect to the largest 
possible power of a non-dictatorial player. It turns out that the considered set of power indices can be 
partitioned into two classes. This may serve as another indication which index to use in a given application.

\medskip


\medskip

\end{abstract}

\maketitle


\section{Introduction}
\label{sec:intro}
\noindent
Consider a community association with four property owners having shares of $50\%$, $26\%$, $15\%$, and $9\%$, respectively.
Assume that decisions are of a simple {\lq\lq}yes{\rq\rq} or {\lq\lq}no{\rq\rq} nature and that the owners decide with a 
two-thirds majority rule. Such a decision environment can be modeled as a weighted game, where the players have non-negative weights 
$w_1,\dots,w_n$. Any subset $S$ of the players, called coalition,  
can adopt a proposal if and only if the sum of their weights $\sum_{i\in S}w_i$ meets or exceeds a given positive 
quota $q$. The collection $[q;w_1,\dots,w_n]$ is then called a weighted game, e.g., $[0.67;0.50,0.26,0.15,0.09]$ in our example. 
Note that those voting weights are often a poor proxy for players' influence. Whenever $S$ is a coalition including the third but 
excluding the fourth player and $T$ is the coalition obtained from exchanging player three by player four, then coalition $S$ can 
bring through a proposal if and only if coalition $T$ can do. So, the third and the fourth player are symmetric in terms 
of their influence on the decision, which is not reflected by the weights. 

The literature has thus introduced several more sophisticated ways of measuring a players' influence in weighted games. 
Unfortunately, different indices can lead to very different predictions. For our example 
we obtain relative power distributions of $(0.50,0.30,0.10,0.10)$, $\frac{1}{12}\cdot(7,3,1,1)$, $(1,0,0,0)$, or $(0.40,0.20,0.20,0.20)$ 
for the Penrose-Banzhaf index, the Shapley-Shubik index, the nucleolus, and the Public Good index, respectively. 

One way to decide which 
power index to choose for a given application is to employ one of the known axiomatizations, see e.g.\ \cite{dubey1975uniqueness}
and \cite{dubey1979mathematical}, and to check which axioms are satisfied. 
Here, we consider the power of the largest player, without full power. It will turn out that the possible values differ 
significantly for different power indices, which may also allow to exclude the 	suitability of certain power indices in a given 
application. Although this theoretical question is quite natural, 
it has not been treated in the 
literature so far.

Another application of our results stems from the so-called 
\emph{inverse power index problem}, see e.g.\ \cite{de2017inverse,koriyama2013optimal,kurz2012inverse, kurz2017democratic}.
It asks for a simple or weighted game $v$ such that the corresponding power distribution (according to a given power index $p$) 
meets a given ideal power distribution $\sigma$ as closely as possible. Since there is only a finite number of different 
weighted or simple games, it is obvious that some power vectors can not be approximated too closely if the number of voters is small. 
\cite{pre05681536} show that there are also vectors that are hard to approximate by the Penrose-Banzhaf index 
of a simple game if most of the mass of the vector is concentrated on a small number of coordinates. Generalizations and impossibility 
results for other power indices have been obtained in \cite{kurz2016inverse}. So, if we know that $p_i(v)=1$ or $p_i(v)\le\lambda$, 
for any simple game $v$, and $\sigma_i$ lies somewhere in the middle of the interval $[\lambda,1]$, then $p(v)$ has a significant distance to 
$\sigma$ provided that $\lambda$ is not 
close to $1$.
 
\section{Preliminaries} 
\label{sec_preliminaries}
\noindent
By $N=\{1,\dots,n\}$ we denote the set of players. A \emph{simple game} is a surjective and monotone mapping $v\colon 2^N\to\{0,1\}$ 
from the set of subsets of $N$ into a binary output $\{0,1\}$. \emph{Monotone} means $v(S)\le v(T)$ for all $\emptyset\subseteq S\subseteq T 
\subseteq N$. The values of this mapping can be interpreted as follows. For each subset $S$ of $N$, called \emph{coalition}, we have 
$v(S)=1$ if the members of $S$ can adopt a proposal even though the members of $N\backslash S$ are against it. If $v(S)=1$ we speak 
of a \emph{winning coalition} and a \emph{losing coalition} otherwise. A winning coalition $S$ is called \emph{minimal} if all of its 
proper subsets are losing. Similarly, a losing coalition $T$ is \emph{maximal} if all of its proper supersets are winning. 
A simple game $v$ is weighted if there exist weights $w_1,\dots,w_n\in\mathbb{R}_{\ge 0}$ and a quota 
$q\in\mathbb{R}_{>0}$ such that $v(S)=1$ exactly if $w(S):=\sum_{i\in S} w_i\ge q$. Two players $i$ and $j$ are called 
\emph{symmetric}, in a given simple game $v$, if $v(S\cup\{i\})=v(S\cup\{j\})$ for all $\emptyset\subseteq S\subseteq N\backslash\{i,j\}$. 
Player $i\in N$ is a \emph{null player} if $v(S)=v(S\cup\{i\})$ for all $\emptyset\subseteq S\subseteq N\backslash\{i\}$, i.e., player 
$i$ is not contained in any minimal winning coalition. A player that is contained in every minimal winning coalition is called a veto 
player. If $\{i\}$ is a winning coalition (note that $\emptyset$ is a losing coalition), then player~$i$ is a \emph{passer}. If additionally 
all other players are null players, then we call player $i$ a \emph{dictator}. 

A power index $p$ is a mapping from the set of simple (or weighted) games on $n$ players into $\mathbb{R}^n$. 
By $p_i(v)$ we denote the $i$th component of $p(v)$, i.e., the power of player~$i$.  As an example 
consider the 
\emph{Shapley-Shubik index}: 
$$
  \operatorname{SSI}_i(v)=\sum_{S\subseteq N\backslash\{i\}} \frac{|S|!\cdot(n-|S|-1)!}{n!}\cdot\left(v(S\cup\{i\})-v(S)\right).
$$
The list of power indices that have been proposed in the literature so far is long. In order to keep the paper compact 
and self-contained, we follow the proposed taxonomy of \cite{kurz2016inverse} 
and refer the reader, e.g., to that paper for more references and details. We call $p$ \emph{positive} if 
$p(v)\in\mathbb{R}_{\ge 0}^n\backslash\{0\}$ and \emph{efficient} if $\sum_{i=1}^n p_i(v)=1$ for all games $v$. For any positive power 
index $p$ we obtain an efficient version by $p_i(v)/\sum_{j=1}^n p_j(v)$. Applying this 
to $\sum_{S\subseteq N\backslash\{i\}}\left(v(S\cup\{i\})-v(S)\right)$ gives the \emph{Penrose-Banzhaf index} $\operatorname{BZI}$. We 
call a coalition $S\cup\{i\}$ \emph{critical} for $i$, if $v(S\cup\{i\})-v(S)=1$. Then player $i$ called \emph{critical}. Note that not 
all players of a critical coalition are critical. 

Instead of counting critical coalitions, we can also count the minimal winning coalitions 
containing a given player~$i$. Normalizing to an efficient version, as above, gives the \emph{Public Good index} $\operatorname{PGI}$.  
The so-called \emph{equal division counting function} gives each relevant player of a counted coalition 
the same share, so that they sum up to one. More concretely 
$$
  \sum_{ \{i\}\subseteq S\subseteq N\,:\, S\text{ is minimal winning}} \frac{1}{|S|} 
$$
gives the non-normalized version of the Deegan-Packel index $\operatorname{DP}$ for player $i$, i.e., it arises from the $\operatorname{PGI}$ 
by equal division. Equal sharing among the critical players of a coalition turns the Penrose-Banzhaf index into the Johnston index 
$\operatorname{Js}$. The 
definition of the \emph{nucleolus} $\operatorname{Nuc}$ is a bit more involved. For a simple game $v$ and a vector $x\in\mathbb{R}^n$ we call 
$e(S,x)=v(S)-x(S)$ the \emph{excess} of $S$ at $x$, where $x(S):=\sum_{i\in S}x_i$. It can be interpreted as quantifying the coalition's 
dissatisfaction and potential opposition to an agreement on allocation $x$. For any fixed $x$ let $S_1, \ldots, S_{2^n}$ be an ordering of 
all coalitions such that the excesses at $x$ are weakly decreasing, and denote these ordered excesses by 
$E(x)=\Big(e(S_k,x)\Big)_{k=1,\ldots, 2^n}$. Vector $x$ is \emph{lexicographically less} than vector $y$ if $E_k(x)<E_k(y)$ for the smallest
component $k$ with $E_k(x)\neq E_k(y)$. The \emph{nucleolus} $x^\star$ of $v$ is then uniquely defined as the lexicographically minimal
vector $x$ with $x(N)\le v(N)=1$, cf.~\cite{0191.49502}. For simple games we automatically have $x^\star\in\mathbb{R}_{\ge 0}^n$ 
and $x^\star(N)=1$. Several authors restrict the definition to imputations, where $x_i^\star\ge v(\{i\})$.

We call a power index $p$ 
\emph{symmetric} if $p_i(v)=p_j(v)$ for symmetric players $i,j$ in $v$. If $p_i(v)=0$ for every null player $i$ of $v$, then we say that 
$p$ satisfies the \emph{null player property}. The six power indices introduced so far are positive, efficient, symmetric, satisfy the 
null player property and are defined for all simple games. 

There are a few other power indices that are just defined for a weighted game $v$ and based on representations. For our initial example 
we have $$[0.67;0.50,0.26,0.15,0.09]=[5;3,2,1,1],$$ i.e., there can be several \emph{representations} of the same weighted game. We can obtain 
power indices for weighted games by averaging over all representations of a certain type. If we restrict to integer weights and quota with 
minimum possible weight sum $\sum_{i=1}^n w_i$, we obtain the \emph{minimum sum representation index} $\operatorname{MSRI}$. We may also average 
over all normalized weight vectors, i.e., over the polyhedron $P^{\text{w}}(v)=\big\{ w\in\mathbb{R}_{\ge 0}^n \,:\,\sum_{i=1}^n w_i=1, 
w(S)\ge w(T)\,\forall \text{ minimal winning $S$ and}$ $\text{all maximal losing $T$}\big\}$. With this the \emph{average weight index} is given by 
$$
  \operatorname{AWI}(v)=\frac{1}{\int_{P^{\text{w}}(v)} \operatorname{d}w}\cdot\left(\int_{P^{\text{w}}(v)} w_1\operatorname{d}w,
  \dots,\int_{P^{\text{w}}(v)} w_n\operatorname{d}w\right).
$$ 
Taking also the quota into account we can consider $P^{\text{r}}(v)=$
$$\left\{ (q,w)\in\mathbb{R}_{\ge 0}^{n+1} : \sum_{i=1}^n w_i=1, 
q\le 1, w(S)\ge q ,\forall \text{ min.win. $\!S$}, w(T)\le q\,\forall\text{ max.los. $\!T$}\right\}$$ and define the \emph{average 
representation index} as 
$$
  \operatorname{ARI}(v)=\frac{1}{\int_{P^{\text{r}}(v)} \operatorname{d}(q,w)}\cdot\left(\int_{P^{\text{r}}(v)} w_1\operatorname{d}(q,w),
  \dots,\int_{P^{\text{r}}(v)} w_n\operatorname{d}(q,w)\right).
$$
All those three representation based power indices are positive, efficient and symmetric. The null voter property is only satisfied for the 
$\operatorname{MSRI}$.
  
\section{Results}
\label{sec_main}
\noindent
For every positive, efficient power index that satisfies the null player property the power of a dictator is exactly one. 
In this case, we speak of full power. So, the largest possible power for a player is $1$ and it is quite natural to ask for the largest possible 
power of a player that is strictly less than $1$. Since the number of simple games is finite for each number $n\in\mathbb{N}$ of players, 
the answer is a well-defined number, which possibly depends on $n$. If $v$ is a simple game with $n\ge 2$ players and player $i$ is not a
dictator, then there exists a player $j\neq i$ that is contained in some minimal winning coalition $S$. Thus, for the Shapley-Shubik, the 
Penrose-Banzhaf, the Public Good index, the Johnston index, and the Deegan-Packel index every player with power $1$ is a dictator. So, 
the condition that player $i$ is not a dictator is equivalent to $p_i(v)<1$ in the following four theorems. Moreover, $n\ge 2$ is implied 
for the number of players. As preparation we observe:
\begin{lemma}
  \label{lemma_characterization_dictator}
  If $v$ is a simple game with player set $N$, $v(N\backslash\{i\})=0$, and $v(\{i\})=1$, then player~$i$ is a dictator.
\end{lemma} 
\begin{proof}
  Since $v(\{i\})=1$ coalition $N$ is not minimal winning. Due to $v(N\backslash\{i\})=0$ player~$i$ is the only player that is 
  contained in a minimal winning coalition, i.e., all other players are null players. So, $\{i\}$ is the unique minimal winning 
  coalition and player~$i$ is a dictator.  
\end{proof}

\begin{theorem}
  \label{thm_largest_ssi}
  For each simple game $v$ on $n\ge 2$ players and each player $i$ that is not a dictator, we have $\operatorname{SSI}_i(v)\le\frac{n-1}{n}$.
\end{theorem}
\begin{proof}
  We compute
  \begin{eqnarray*}
    \operatorname{SSI}_i(v)&=&\frac{1}{n!}\cdot\sum_{S\subseteq N\backslash\{i\}} |S|!\cdot(n-|S|-1)!\cdot\left(v(S\cup\{i\})-v(S)\right)\\
    &\le& \frac{1}{n!}\cdot\sum_{S\subseteq N\backslash\{i\}} |S|!\cdot(n-|S|-1)!
    =\operatorname{SSI}_1([1;1,0,\dots,0])=1.
  \end{eqnarray*}
  Since either $v(N\backslash\{i\})=1$ or $v(\{i\})=0$, due to Lemma~\ref{lemma_characterization_dictator}, we have 
  $\operatorname{SSI}_i(v)\le 1-\frac{(n-1)!\cdot 1!}{n!}=\frac{n-1}{n}$. 
\end{proof}

We remark that the upper bound is met for $v=[n-1;n-1,1,\dots,1]$ and that it approaches $1$ as $n$ tends to infinity.

\begin{lemma}
  \label{lemma_special_weighted_game}
  For $n\ge 2$ and $v=[n-1;n-1,1,\dots,1]$ we have $\operatorname{SSI}(v)=\frac{1}{n(n-1)}\cdot \left((n-1)^2,1,\dots,1\right)$,  
  $\operatorname{BZI}(v)=\frac{1}{2^{n-1}+n-2}\cdot \left(2^{n-1}-1,1,\dots,1\right)$, and 
  $\operatorname{Js}(v)=\frac{1}{2^{n-1}}\cdot \left(2^{n-1}-1,\frac{1}{n-1},\dots,\frac{1}{n-1}\right)$.
\end{lemma}
\begin{proof}
  For all $2^{n-1}-1$ coalitions $\emptyset\subseteq S\subsetneq N\backslash\{i\}$ we have $v(S\cup\{i\})-v(S)=1$, while 
  $v(N)-v(N\backslash\{i\})=0$. For any player $j\in N\backslash\{i\}$ the only coalition $S\subseteq N\backslash\{j\}$ with 
  $v(S\cup\{j\})-v(S)=1$ is given by $S=N\backslash\{i,j\}$. 
\end{proof}

\begin{theorem}
  \label{thm_largest_bzi}
  For each simple game $v$ on $n\ge 2$ players and each player $i$ that is not a dictator, we have 
  $\operatorname{BZI}_i(v)\le\frac{2^{n-1}-1}{2^{n-1}+n-2}$.
\end{theorem}
\begin{proof}
  Let $\psi_j(v)=\sum_{S\subseteq N\backslash\{j\}}\left(v(S\cup\{j\})-v(S)\right)$, i.e., the Penrose-Banzhaf index before 
  normalization. Due to Lemma~\ref{lemma_characterization_dictator} we have $v(N\backslash\{i\})=1$ or $v(\{i\})=0$, so that 
  $\psi_i(v)\le 2^{n-1}-1$. Assuming that $v$ contains no null player, we have $\psi_j(v)\ge 1$ for all $j\in N\backslash\{i\}$ 
  since $j$ is contained in at least one minimal winning coalition. Thus, $\operatorname{BZI}_i(v)\le \frac{2^{n-1}-1}{2^{n-1}+n-2}$. 
  If $v$ contains at least one null player $h$. Let $v'$ be the simple game with player set $N\backslash \{h\}$ defined by 
  $v'(S)=v(S)$ for all $\emptyset\subseteq S\subseteq N\backslash\{h\}$. For any player $j\in N\backslash\{h\}$ and any coalition 
  $S\subseteq N\backslash\{j,h\}$ we have $v(S\cup\{j\})-v(S)=1$ if and only if $v(S\cup\{j,h\})-v(S\cup\{h\})=1$, so that 
  $\psi_j(v)=2\psi_j(v')$. By induction we get $\operatorname{BZI}_i(v)\le \frac{2^{n-2}-1}{2^{n-2}+n-3}=1-\frac{n-2}{2^{n-2}+n-3}<
  1-\frac{n-1}{2^{n-1}+n-2}=\frac{2^{n-1}-1}{2^{n-1}+n-2}$ for all $n\ge 3$.  
  (Note that any simple game with $n\le 2$ players either contains a dictator or no null player at all.)    
\end{proof}

Similarly, we obtain:
\begin{theorem}
  \label{thm_largest_js}
  For each simple game $v$ on $n\ge 2$ players and each player $i$ that is not a dictator, we have 
  $\operatorname{Js}_i(v)\le\frac{2^{n-1}-1}{2^{n-1}}$.
\end{theorem}
\begin{proof}
  Let 
  $
    \psi'_j(v)=\sum\limits_{S\subseteq N\backslash\{j\}} \!\!\!\left(v(S\cup\{j\})-v(S)\right) / \left(\#\text{ of critical players of $S$}\right)
  $,
  i.e., the Johnston index before normalization. Due to Lemma~\ref{lemma_characterization_dictator} we have $v(N\backslash\{i\})=1$ 
  or $v(\{i\})=0$, so that $\psi'_i(v)\le 2^{n-1}-1$. Since there exists a minimal winning coalition $S\subseteq N\backslash\{i\}$, 
  we have $\sum_{j\in N\backslash\{i\}} \psi'_j(v)\ge \sum_{j\in S} \psi'_j(v)\ge 1$, so that $\operatorname{Js}_i(v)\le\frac{2^{n-1}-1}{2^{n-1}}$.     
\end{proof}

Again, the upper bound of Theorem~\ref{thm_largest_bzi} and Theorem~\ref{thm_largest_js} is met for $v=[n-1;n-1,1,\dots,1]$, see 
Lemma~\ref{lemma_special_weighted_game}, and approaches $1$ as $n$ tends to infinity.

\begin{theorem}
  \label{thm_largest_pgi}
  For each simple game $v$ on $n\ge 2$ players and each player $i$ that is not a dictator, we have 
  $\operatorname{PGI}_i(v)\le\frac{1}{2}$ and $\operatorname{DP}_i(v)\le\frac{1}{2}$.
\end{theorem}
\begin{proof}
  If $\{i\}$ is a winning coalition, then it is the only minimal winning coalition containing player~$i$. Due to 
  Lemma~\ref{lemma_characterization_dictator}, we have $v(N\backslash\{i\})=1$, so that there exists a minimal winning 
  coalition $S\subseteq N\backslash\{i\}$. Thus, $\operatorname{PGI}_i(v)\le \frac{1}{2}$. If $\{i\}$ is a losing coalition, then 
  either $i$ is a null player or any minimal winning coalition containing player $i$ has a cardinality of at least $2$, so that 
  $\operatorname{PGI}_i(v)\le \frac{1}{2}$. The same reasoning 
  applies to the Deegan-Packel index.  
\end{proof}

The upper bound is attained for $[1;1,1,0,\dots,0]$ and $[2;1,1,0,\dots,0]$. 

We remark that for \emph{complete simple games}, a class in between weighted and simple games, a power index (called 
\emph{Shift index}) based on counting so-called \emph{shift-minimal} winning coalitions instead of minimal winning coalitions 
and the corresponding equal division version (called \emph{Shift Deegan-Packel index index}) can be defined, see e.g.\ 
\cite{kurz2016inverse} and the references therein. The result of Theorem~\ref{thm_largest_pgi} and its proof directly transfer. 

\medskip

The nucleolus is special. 
Of course the nucleolus also attributes power $1$ to a dictator. However, there are also non-dictatorial simple games where one player 
gets a nucleolus power of $1$. It is well known that the nucleolus of a simple game with $k\ge 1$ veto players assigns 
$\frac{1}{k}$ to the veto players and zero to the remaining players. For $k=1$ we obtain all simple games with a player having 
full nucleolus power.  

\begin{proposition}
  \label{prop_characterization_nuc_1}
  If $v$ is a simple game and $i$ be a player with $\operatorname{Nuc}_i(v)=1$, then $i$ is the unique veto player.
\end{proposition}
\begin{proof}
  If player $i$ is the unique veto player, then $\operatorname{Nuc}_i(v)=1$. If another player is the unique veto player, then 
  $\operatorname{Nuc}_i(v)=0$. If there are at least $2$ veto players, then $\operatorname{Nuc}_i(v)\le \frac{1}{2}$. Thus, we can 
  assume that $v$ contains no veto players, so that there exists a winning coalition $S\subseteq N\backslash\{i\}$, and $\operatorname{Nuc}_i(v)=1$. 
  Abbreviating $\operatorname{Nuc}(v)$ by $x^\star$, we have $x^{\star}(S)=0$, so that $\max_{C\subseteq N} e(C,x^\star)\ge e(S,x^\star)=1$.
  However, $\max_{C\subseteq N} e\!\left(C,\frac{1}{n}\cdot(1,\dots,1)\right)\le\frac{n-1}{n}<1$, which is a contradiction.   
\end{proof}

\begin{theorem}
  \label{thm_largest_nucleolus}
  Let $v$ be a simple game and $i$ be a player with $\operatorname{Nuc}_i(v)<1$, then $\operatorname{Nuc}_i(v)\le \frac{1}{2}$. 
\end{theorem}
\begin{proof}
  As in the proof of Proposition~\ref{prop_characterization_nuc_1} we can assume that $v$ contains no veto player, choose a winning 
  coalition $S\subseteq N\backslash\{i\}$, and introduce the abbreviation $x^\star=\operatorname{Nuc}(v)$. Note that we have 
  $n\ge 2$ players. Assume $x^\star_i>\frac{1}{2}$ 
  and set $\varepsilon:=x^\star_i-\frac{1}{2}>0$. Let $T$ be a winning coalition with minimal $x^\star(T)$. Since 
  $x^\star(S)\le 1-x^\star_i= \frac{1}{2}-\varepsilon$, we have $x^\star(T)\le\frac{1}{2}-\varepsilon$. Now we define $x_i=\frac{1}{2}$ and 
  $x_j=x^\star_j+\frac{\varepsilon}{n-1}$ for all $j\in N\backslash\{i\}$. For each winning coalition
  $W$ with $i\in W$ we have $x(W)\ge \frac{1}{2}>x^\star(T)$ and for each winning coalition $W'$ with $i\notin W'$ we have 
  $x(W')=x^\star(W')+|W'|\cdot\frac{\varepsilon}{n-1}\ge x^\star(T)+\frac{\varepsilon}{n-1}>x^\star(T)$, 
  which is a contradiction to the minimality 
  of the nucleolus.
\end{proof}

The upper bound is, e.g., attained for simple games with exactly two veto players. However, there are many other examples.

\begin{theorem}
  \label{thm_largest_msri}
  For each weighted game $v$ on $n\ge 2$ players and each player $i$ that is not a dictator, we have 
  $\operatorname{MSRI}_i(v)\le\frac{1}{2}$.
\end{theorem}
\begin{proof}
  Suppose $(q;w)\in\mathbb{N}^{n+1}$ is a representation of $v$ with minimum $\sum_{i=1}^n w_i$ and set $r:=\sum_{j\in N\backslash\{i\}} w_j\ge 1$. 
  Assume $w_i\ge r+1$. If $q\le r$, then we can replace $w_i$ by $r$ and obtain a representation with a smaller sum, a contradiction. 
  If $q\ge r+1$, then player~$i$ is a veto player. Note that $q\ge w_i+1\ge r+2$, since otherwise player~$i$ is a dictator. However, 
  reducing $q$ and $w_i$ by $1$ gives a representation with a smaller sum, again a contradiction. Thus, we have $w_i\le r$ for every minimum 
  sum representation, so that $\operatorname{MSRI}_i(v)\le\frac{1}{2}$.    
\end{proof}
If player~$1$ is a dictator in a weighted game $v$, then the unique minimum sum representation is given by $[1;1,0,\dots]$, so that 
$\operatorname{MSRI}_1(v)=1$. The upper bound is met by $[k;k,1,\dots,1,0,\dots,0]$, with $k\ge 2$ players of weight $1$.

For the average weight and the average representation index even a dictator does not get a power of $1$ for $n\ge 2$ players.
\begin{theorem}
  \label{thm_largest_average}
  For each weighted game $v$ on $n\ge 1$ players and each player $i$, we have 
  $\operatorname{AWI}_i(v)\le\frac{n+1}{2n}$ and $\operatorname{ARI}_i(v)\le\frac{n+3}{2(n+1)}$.
\end{theorem}
\begin{proof}
  The statement is true for $n=1$, so that we assume $n\ge 2$.
  The maximum values are clearly attained for a dictator. So, for $v=[1;1,0,\dots,0]$ we have 
  $$
    \int_{P^{\text{w}}(v)} \operatorname{d}w=\int_{\frac{1}{2}}^{1}\int_{0}^{1-w_1}\dots \int_{0}^{1-w_1-\dots-w_{n-2}}  
    \operatorname{d}w_{n-1} \dots\operatorname{d}w_2 \operatorname{d}w_1. 
  $$
  Since $\int_{0}^y \frac{(y-x)^k}{k!}\operatorname{d}x=\int_{0}^y \frac{x^k}{k!}\operatorname{d}x=\frac{y^{k+1}}{(k+1)!}$ for 
  all $k\in\mathbb{N}$, we recursively compute
  $$
    \int_{P^{\text{w}}(v)} \operatorname{d}w=\int_{\frac{1}{2}}^{1} \frac{(1-w_1)^{n-2}}{(n-2)!}\operatorname{d}w_1
    =\int_{0}^{\frac{1}{2}} \frac{w_1^{n-2}}{(n-2)!}\operatorname{d}w_1=\frac{1}{2^{n-1}\cdot(n-1)!} 
  $$
  and
  $$
    \int_{P^{\text{w}}(v)} w_1\operatorname{d}w=\int_{\frac{1}{2}}^{1} \frac{w_1(1-w_1)^{n-2}}{(n-2)!}\operatorname{d}w_1
    =\int_{0}^{\frac{1}{2}} \frac{(1-w_1)w_1^{n-2}}{(n-2)!}\operatorname{d}w_1=\frac{n+1}{2^{n}\cdot n!},
  $$
  so that $\operatorname{AWI}_1(v)=\frac{n+1}{2n}$. Similarly, we have 
  $$
    \int_{P^{\text{r}}(v)} \operatorname{d}w=\int_{0}^{1}\int_{\max\{q,1-q\}}^1\int_{0}^{1-w_1}\!\!\!\!\dots \int_{0}^{1-w_1-\dots-w_{n-2}}  
    \operatorname{d}w_{n-1} \dots\operatorname{d}w_2 \operatorname{d}w_1 \operatorname{d}q, 
  $$ 
  so that
  $$
    \int_{P^{\text{r}}(v)} \operatorname{d}w 
    =2\int_{\frac{1}{2}}^{1}\int_{q}^1 \frac{(1-w_1)^{n-2}}{(n-2)!}\operatorname{d}w_1 \operatorname{d}q
    =\frac{1}{2^{n-1}\cdot n!}
  $$
  and
  $$
    \int_{P^{\text{r}}(v)} w_1\operatorname{d}w=2\int_{\frac{1}{2}}^{1}\int_{q}^1 \frac{w_1(1-w_1)^{n-2}}{(n-2)!}\operatorname{d}w_1 \operatorname{d}q
    =\frac{n+3}{2^{n}\cdot(n+1)!}.
  $$
  Thus, $\operatorname{ARI}_1(v)=\frac{n+3}{2(n+1)}$.
\end{proof}
We remark that the upper bound approaches $\frac{1}{2}$ when $n$ tends to infinity. 

\section{Conclusion}
\label{sec_conclusion}
\noindent
For several power indices $p$ we have determined tight upper bounds $\alpha_p(n)$ for $p_i(v)$, where $p_i(v)<1$ and $v$ is either a 
simple or a weighted game on $n$ players. The considered power indices fall into two classes: If $n$ tends to infinity, then 
$\alpha_p(n)$ approaches either $\frac{1}{2}$ or $1$. 

More precisely, this implies:
\begin{theorem}
  \label{thm_summary}
  Let $p$ be one of the power indices $\operatorname{PGI}$, $\operatorname{DP}$, $\operatorname{Nuc}$, $\operatorname{MSRI}$, $\operatorname{AWI}$, 
  or $\operatorname{ARI}$. For each $\varepsilon>0$ there exists an integer $N(\varepsilon)$ such that either $p_i(v)=1$ or 
  $p_i(v)\le \frac{1}{2}+\varepsilon$ for each simple game $v$ on $n\ge N(\varepsilon)$ players and an arbitrary player $1\le i\le n$.  
\end{theorem} 
For the first four mentioned power indices we may even choose $\varepsilon=0$ and $N(\varepsilon)=1$. Such a {\lq\lq}hole{\rq\rq} in 
the space of possible power values does not occur for the power indices $\operatorname{SSI}$, $\operatorname{BZI}$, or $\operatorname{Js}$. 
If it is essential in an application to differentiate the influence of a {\lq\lq}large{\rq\rq} player in different simple games, then 
the power indices from Theorem~\ref{thm_summary} disqualify.   

A direct implication for the inverse power index problem is given by:
\begin{proposition}
  Let $p$ be a power index and $\sigma\in[0,1]^n$ with $\sum_{i=1}^n \sigma_i=1$ and $\alpha_p(n)\le \sigma_j\le 1$ for some 
  player $1\le j\le n$. Then, $\Vert \sigma-p(v)\Vert_1\ge \min\{1-\sigma_j,\sigma_j-\alpha_p(n)\}$ for ever simple game $v$ on $n$ players.
\end{proposition}

For example, the desired power distribution $\sigma=(0.75,0.25,0,\dots,0)$ cannot be approximated with $\Vert\cdot\Vert_1$-distance 
strictly less than $\frac{1}{4}$ by the power distribution of a simple game according to one of the first four power indices 
from Theorem~\ref{thm_summary}. For similar results for the Penrose-Banzhaf index see \cite{pre05681536} and \cite{kurz2016inverse}.

Of course we may ask for upper bounds for $p_i(v)$, where $p_i(v)<\alpha_p(n)$, and 
so on. For the Shapley-Shubik index and $n\ge 3$ we conjecture that the next two upper bounds are given by 
$\frac{n-2}{n-1}=\frac{n-1}{n}-\frac{1}{n(n-1)}$ and $\frac{n^2-2n-1}{n(n-1)}=\frac{n-1}{n}-\frac{2}{n(n-1)}$. However, it seems that those 
values approach the same limit as $\alpha_p(n)$ in any case. In other words, for the power of the largest player the only gap, except 
for small values, that does not vanish when $n$ increases, is given by $(\alpha_p(n),1)$. Similar questions can be asked for the 
smallest player that is not a null-player or the second largest player and so on.


\begin{thebibliography}{9}
\expandafter\ifx\csname natexlab\endcsname\relax\def\natexlab#1{#1}\fi
\expandafter\ifx\csname url\endcsname\relax
  \def\url#1{\texttt{#1}}\fi
\expandafter\ifx\csname urlprefix\endcsname\relax\def\urlprefix{URL }\fi

\bibitem{pre05681536}
Alon, N., Edelman, P., 2010. The inverse {B}anzhaf problem. Social Choice and
  Welfare 34~(3), 371--377.

\bibitem{de2017inverse}
De, A., Diakonikolas, I., Servedio, R.~A., 2017. The inverse {S}hapley value
  problem. Games and Economic Behavior 105, 122--147.

\bibitem{dubey1975uniqueness}
Dubey, P., 1975. On the uniqueness of the {S}hapley value. International
  Journal of Game Theory 4~(3), 131--139.

\bibitem{dubey1979mathematical}
Dubey, P., Shapley, L., 1979. Mathematical properties of the {B}anzhaf power
  index. Mathematics of Operations Research 4~(2), 99--131.

\bibitem{koriyama2013optimal}
Koriyama, Y., Mac{\'e}, A., Treibich, R., Laslier, J.-F., 2013. Optimal
  apportionment. Journal of Political Economy 121~(3), 584--608.

\bibitem{kurz2012inverse}
Kurz, S., 2012. On the inverse power index problem. Optimization 61~(8),
  989--1011.

\bibitem{kurz2016inverse}
Kurz, S., 2016. The inverse problem for power distributions in committees.
  Social Choice and Welfare 47~(1), 65--88.

\bibitem{kurz2017democratic}
Kurz, S., Maaser, N., Napel, S., 2017. On the democratic weights of nations.
  Journal of Political Economy 125~(5), 1599--1634.

\bibitem{0191.49502}
Schmeidler, D., 1969. The nucleolus of a characteristic function game. SIAM
  Journal on Applied Mathematics 17, 1163--1170.

\end{thebibliography}

\end{document}